\documentclass[10pt]{amsart}
\usepackage{epsfig,url}
\usepackage{mathdots}

\newtheorem{theorem}{Theorem}[section]
\newtheorem{lemma}[theorem]{Lemma}

\newtheorem{corollary}[theorem]{Corollary}

\theoremstyle{definition}

\newtheorem{note}[theorem]{Note}

\theoremstyle{remark}

\begin{document}

\title[Powers of Bessel functions]
{On polynomials connected to powers of Bessel functions}

\author[]{Victor H. Moll}
\address{Department of Mathematics,
Tulane University, New Orleans, LA 70118}
\email{vhm@tulane.edu}

\author[]{Christophe Vignat}
\address{Department of Mathematics,
Tulane University, New Orleans, LA 70118}
\email{cvignat@tulane.edu}


\date{\today}


\begin{abstract}
The series expansion of a power of the modified Bessel function of the first 
kind is studied. This expansion involves a family of polynomials introduced 
by C. Bender et al. New results on these polynomials established here 
include recurrences in terms of Bell polynomials 
evaluated at values of the Bessel zeta function. A probabilistic version of an 
identity of Euler yields additional recurrences. Connections to 
the umbral formalism on Bessel functions introduced by 
Cholewinski are established.
\end{abstract}

\maketitle

\newcommand{\ba}{\begin{eqnarray}}
\newcommand{\ea}{\end{eqnarray}}
\newcommand{\ift}{\int_{0}^{\infty}}
\newcommand{\nn}{\nonumber}
\newcommand{\no}{\noindent}
\newcommand{\lf}{\left\lfloor}
\newcommand{\rf}{\right\rfloor}
\newcommand{\realpart}{\mathop{\rm Re}\nolimits}
\newcommand{\imagpart}{\mathop{\rm Im}\nolimits}

\newcommand{\op}[1]{\ensuremath{\operatorname{#1}}}
\newcommand{\pFq}[5]{\ensuremath{{}_{#1}F_{#2} \left( \genfrac{}{}{0pt}{}{#3}
{#4} \bigg| {#5} \right)}}

\newtheorem{Definition}{\bf Definition}[section]
\newtheorem{Thm}[Definition]{\bf Theorem}
\newtheorem{Example}[Definition]{\bf Example}
\newtheorem{Lem}[Definition]{\bf Lemma}
\newtheorem{Cor}[Definition]{\bf Corollary}
\newtheorem{Prop}[Definition]{\bf Proposition}
\numberwithin{equation}{section}

\section{Introduction}
\label{sec-intro}

Starting with the power series expansion
\begin{equation}
f(z) = \sum_{n=0}^{\infty} \frac{1}{a_{n}} \frac{z^{n}}{n!}, \quad 
a_{n} \neq 0, 
\end{equation}
\noindent
the corresponding expansion for $\left[f(z)\right]^{r}, \, r \in \mathbb{N}$
\begin{equation}
\left[ f(z) \right]^{r} = \sum_{n=0}^{\infty} A_{n}(r) \frac{1}{a_{n}} 
\frac{z^{n}}{n!}
\label{def-A}
\end{equation}
\noindent 
defines the coefficients $A_{n}(r)$. For instance, $A_{0}(r) = 1/a_{0}^{r-1}$. 
The identity $\left[ f(z) \right]^{r+1} = \left[f(z) \right]^{r} \times f(z)$ 
produces
\begin{equation}
A_{n}(r+1) = \sum_{j=0}^{n} 
\left( n \atop j \right)_{\mathbf{a}} 
A_{j}(r),
\label{recu-A}
\end{equation}
\noindent
where the generalized binomial coefficients $\left( n \atop j \right)_{\mathbf{a}}$ are defined as 
\begin{equation}
\left( {n \atop j} \right)_{\mathbf{a}} = 
\binom{n}{j} \frac{a_{n}}{a_{j}a_{n-j}}.
\end{equation}

The example motivating the results presented here comes from work by 
C. Bender at al \cite{bender-2003a} on the normalized Bessel function
\begin{equation}
\tilde{I}_{\nu}(z) = \frac{\nu! 2^{\nu}}{z^{\nu}} I_{\nu}(z)
= \sum_{k=0}^{\infty} \frac{\nu!}{k! (k + \nu)!}
\left( \frac{z}{2} \right)^{2k+\nu}.
\label{bess-nor00}
\end{equation}
\noindent
Here 
\begin{equation}
I_{\nu}(z) = \sum_{k=0}^{\infty} \frac{1}{k! (k + \nu)!} 
\left( \frac{z}{2} \right)^{2k+\nu},
\end{equation}
is the \textit{modified Bessel function of the first kind}. The main result in 
\cite{bender-2003a} is stated next.

\begin{theorem}[Bender et al.]
For $r \in \mathbb{N}$, the power series expansion 
\begin{equation}
\left[ \tilde{I}_{\nu}(z) \right]^{r} = 
\sum_{k=0}^{\infty} \frac{\nu!}{k! (k+ \nu)!} B_{k}^{(\nu)}(r) 
\left( \frac{z}{2} \right)^{2k},
\label{expansion-1}
\end{equation}
\noindent 
holds for polynomials $B_{k}^{(\nu)}(r)$, determined recursively by 
\begin{equation}
B_{k}^{(\nu)}(r) = 
r \frac{\nu+k}{\nu+1} B_{k-1}^{(\nu)}(r) + 
\sum_{j=2}^{k} \frac{b_{j}(\nu)}{k} \frac{(\nu+1)!}{(\nu+1+j)!} 
\binom{\nu+k}{j} B_{k-j}^{(\nu)}(r)
\label{B-def}
\end{equation}
\noindent
with initial condition 
$B_{0}^{(\nu)}(r) = 1$. The sequence $b_{j}(\nu)$ has the generating 
function 
\begin{equation}
\sum_{k=1}^{\infty} \frac{b_{k}(\nu)}{(\nu+k)! (k-1)!}x^{k} = 
\left( \frac{\sqrt{x}}{\nu+1} \frac{I_{\nu}(2 \sqrt{x})}{I_{\nu+1}(2 \sqrt{x})}
- 2 \right) \frac{x}{(\nu+1)!}.
\end{equation}
\end{theorem}

The goal of this work is to present an alternative approach to the 
expansion \eqref{expansion-1}. Section \ref{sec-polynomials-direct} uses the 
identity \eqref{recu-A} to derive an expression for 
$B_{k}^{(\nu)}(r+1)$ in terms of $B_{j}^{(\nu)}(r)$. Section 3 identifies
$B_{k}^{(\nu)}(r)$ as Bell polynomials and produces a new recurrence 
involving the Bessel zeta function. Section 4 uses an identity of Euler to 
derive a second recurrence for
$B_{k}^{(\nu)}(r)$. Section 5 shows that these polynomials are of binomial 
type and a link with Cholewinski's theory of Bessel functions is 
exhibited. The last section is dedicated to arithmetic properties 
of some sequences related to powers of Bessel functions.

\smallskip

\noindent
The notation $\nu! = \Gamma(\nu+1)$ is employed throughout. 

\section{A first identity for $B_{n}^{(\nu)}(r)$.}
\label{sec-polynomials-direct}

The expansion of $\tilde{I}_{\nu}(z)$, written as 
\begin{equation}
\tilde{I}_{\nu}(\sqrt{z}) = \sum_{k=0}^{\infty} \frac{1}{a_{k}^{(\nu)}} 
\frac{z^{k}}{k!}
\end{equation}
\noindent
with 
\begin{equation}
a_{k}^{(\nu)} = \frac{(k + \nu)!}{\nu!} 2^{2k},
\end{equation}
\noindent
shows that 
the coefficients $A_{k}(r)$ in \eqref{def-A} are precisely the polynomials 
$B_{k}^{(\nu)}(r)$ in \eqref{expansion-1}.  The 
fact that $B_{n}^{(\nu)}(r)$ are polynomials, not a priori clear from 
their definition, follows from \eqref{recu-A}.

\begin{theorem}
The functions $B_{n}^{(\nu)}(r)$ satisfy
\begin{equation}
B_{n}^{(\nu)}(r+1) = \sum_{j=0}^{n} \binom{n}{j} 
\frac{(n+\nu)! \, \nu!}{(\nu+j)! \, (n-j+\nu)!} B_{j}^{(\nu)}(r).
\label{first-relation}
\end{equation}
\end{theorem}

\begin{corollary}
\label{they-poly}
The coefficients $B_{n}^{(\nu)}(r)$ are polynomials in $r$ of degree $n$.
\end{corollary}
\begin{proof}
Proceed by induction on $n$, writing \eqref{first-relation} as 
\begin{equation}
B_{n}^{(\nu)}(r+1) -B_{n}^{(\nu)}(r) = \sum_{j=0}^{n-1} \binom{n}{j} 
\frac{(n+\nu)! \, \nu!}{(\nu+j)! \, (n-j+\nu)!} B_{j}^{(\nu)}(r).
\end{equation}
\noindent
The result follows by differentiating $n+1$ times with respect to $r$.
\end{proof}
First examples of these polynomials are:
\[
B_{0}^{(\nu)}(r)=1,\,\,B_{1}^{(\nu)}(r)=r,\,\,B_{2}^{(\nu)}(r) = \frac{\nu+2}{\nu+1}r^2-\frac{1}{\nu+1}r
\]
\[
B_{3}^{(\nu)}(r)=\frac{\left(\nu+3\right)\left(\nu+2\right)}{\left(\nu+1\right)^2}r^3-3\frac{\nu+3}{\left(\nu+1\right)^2 }r^2+\frac{4}{\left(\nu+1\right)^2}r
\].

\begin{note}
The identity \eqref{first-relation}, established first for $r \in 
\mathbb{N}$, naturally holds for $r \in \mathbb{R}$. The same principle applies
to the recurrences involving $B_{n}^{(\nu)}(r)$ presented below. 
\end{note}

\section{An identity for $B_{n}^{(\nu)}(r)$ in terms of the Bessel zeta 
function.}
\label{sec-bessel-zeta}

This section provides an expression for the polynomials
$B_{n}^{(\nu)}(r)$, starting from
\begin{equation}
\left[ \tilde{I}_{\nu}(z) \right]^{r} = 
\exp \left[ r \, \log \tilde{I}_{\nu}(z) \right]
\label{bessel-exp1}
\end{equation}
\noindent
and using the Hadamard factorization
\begin{equation}
\tilde{I}_{\nu}(z) = \prod_{k=1}^{\infty} 
\left( 1 + \frac{z^{2}}{j_{\nu,k}^{2}} \right).
\label{factor-bessel}
\end{equation}
\noindent
Here $\{ j_{\nu,k} \}_{k\geq 1}$ is the sequence of zeros of
$\tilde{J}_{\nu}(z) = \tilde{I}_{\nu}(\imath z)$.

The exponential of a power series is computed by
\begin{equation}
\exp \left[ \sum_{p=1}^{\infty} a_{p} \,  \frac{z^{p}}{p!} \right] = 
\sum_{p=0}^{\infty} \mathfrak{B}_{n}(a_{1},\cdots,a_{n})\, \frac{z^{n}}{n!},
\label{bell-1}
\end{equation}
\noindent
where  $\mathfrak{B}_{n}(a_{1},\cdots,a_{n})$ is the $n$-th \textit{complete 
Bell polynomial}. Details
appear in \cite[p.173]{riordan-1968a}. The first few examples are 
\begin{multline*}
\mathfrak{B}_{0} = 1, \,  \mathfrak{B}_{1}(a_{1}) = a_{1},  \, 
\mathfrak{B}_{2}(a_{1},a_{2}) = a_{1}^{2} + a_{2}, \, 
\mathfrak{B}_{3}(a_{1},a_{2},a_{3}) = a_{1}^{3} + 3a_{1}a_{2}+a_{3}.
\end{multline*}
\noindent
From \eqref{factor-bessel}, it follows that 
\begin{equation}
\log \tilde{I}_{\nu}(z) = \sum_{p=1}^{\infty} \frac{(-1)^{p+1}}{p} 
\zeta_{\nu}(2p) z^{2p},
\label{prod-111}
\end{equation}
\noindent
where $\zeta_{\nu}(p)$ is the \textit{Bessel zeta function} 
\cite{kishore-1964a}, defined by
\begin{equation}
\zeta_{\nu}(p) = \sum_{k=1}^{\infty} \frac{1}{j_{\nu,k}^{p}},\,\,p>1.
\end{equation}

A direct application of \eqref{bell-1} yields the next result.

\begin{theorem}
Define 
\begin{equation}
a_{n} = a_{n}(r) = (-1)^{n-1} (n-1)! \zeta_{\nu}(2n) r. 
\label{def-a}
\end{equation}
\noindent 
Then $B_{n}^{(\nu)}(r)$ is given by
\begin{equation}
B_{n}^{(\nu)}(r) = 2^{2n} \frac{(n+ \nu)!}{\nu!} 
\mathfrak{B}_{n}(a_{1}(r),\cdots,a_{n}(r)).
\label{bn-bell}
\end{equation}
\end{theorem}

A recurrence for $B_{n}^{(\nu)}(r)$ is now obtained from the classical identity \cite[p.174]{riordan-1968a}
\begin{equation}
\mathfrak{B}_{n}(a_{1},\cdots, a_{n}) = \sum_{k=0}^{n-1} 
\binom{n-1}{k} a_{k+1}\mathfrak{B}_{n-1-k}(a_{1},\cdots, a_{n-1-k})
\label{recu-bell}
\end{equation}
\noindent
for the complete Bell polynomials.

\begin{theorem}
\label{recu-001}
The polynomials $B_{n}^{(\nu)}(r)$ satisfy the recurrence 
\begin{equation*}
B_{n}^{(\nu)}(r) = r (\nu + n)! 
\sum_{k=0}^{n-1} \binom{n-1}{k} \frac{(-1)^{k} \, k!}
{(\nu + n - k -1 )!} 2^{2k+2} \zeta_{\nu}(2k+2) 
B_{n-1-k}^{(\nu)}(r).
\end{equation*}
\noindent
with initial condition $B_{0}^{(\nu)}(r) = 1$.
\end{theorem}
\begin{proof}
Simply replace \eqref{bn-bell} in \eqref{recu-bell}.
\end{proof}

\begin{note}
The recurrence above provides a new proof of Corollary \ref{they-poly}.  
However, it 
is not easy to use, since there is no 
explicit expression for the coefficients. These involve the values 
$\zeta_{\nu}(2k)$, that can be obtained from 
\begin{equation}
\label{recu-zeta}
(n + \nu) \zeta_{\nu}(2n) = 
\sum_{r=1}^{n-1} \zeta_{\nu}(2r) \zeta_{\nu}(2n-2r),
\end{equation}
\noindent
established in \cite{elizalde-1993a}. The initial condition 
$\zeta_{\nu}(2) = \frac{1}{4(\nu+1)}$ shows that 
$\zeta_{\nu}(2n)$ is 
a rational function of $\nu$. These have recently appeared in 
connection with Narayana polynomials 
\cite{amdeberhan-2013a}. The first 
few are 
\begin{equation*}
\zeta_{\nu}(4) = \frac{1}{16(\nu+1)^{3}}, \, 
\zeta_{\nu}(6) = \frac{1}{16(\nu+1)^{4}(2 \nu+3)}, \, 
\zeta_{\nu}(8) = \frac{10 \nu + 11}{256(\nu+1)^{6}(2 \nu^{2} + 7 \nu + 6)}.
\end{equation*}
\noindent
A more explicit recurrence for $B_{n}^{(\nu)}(r)$ is given in the next 
section.
\end{note}

\section{A second recurrence for the polynomials $B_{n}^{(\nu)}(r)$.}
\label{sec-recurrence2}

This section describes a new recurrence for the polynomials 
$B_{n}^{(\nu)}(r)$. The proof is based on a probabilistic interpretation 
of a beautiful result of L. Euler, recently used by 
A. Baricz \cite{baricz-2010a} to discuss properties for powers of 
Bessel functions. 

\begin{theorem}[Euler]
Let 
\begin{equation}
f(x) = \sum_{n=0}^{\infty} c_{n}x^{n}
\end{equation}
\noindent
and assume $c_{0} \neq 0$. The coefficients $d_{n}$ in 
\begin{equation}
\left( f(x) \right)^{r} = \sum_{n=0}^{\infty} d_{n}x^{n}
\end{equation}
\noindent
are given by the recurrence 
\begin{equation}
d_{n} = \frac{1}{c_{0}} \sum_{k=1}^{n} \left( \frac{k}{n} (r+1) - 1 
\right) c_{k} d_{n-k}, \,\,n\ge1
\label{euler-000}
\end{equation}
\noindent
with initial condition $d_{0} = c_{0}^{r}$.
\end{theorem}

A probabilistic counterpart is stated next.  The proof presented here does 
not require to have a priori knowledge of Euler's formula \eqref{euler-000}.

\begin{lemma}
Assume $X_{1}, \, X_{2}, \, \cdots, X_{r+1}$ are $r+1$ independent 
identically distributed random 
variables with vanishing odd moments; that is,
\begin{equation}
\mathbb{E} X_{i}^{2k+1} = 0, \text{ for all } k \in \mathbb{N}.
\end{equation}
\noindent
Then 
\begin{equation}
\mathbb{E} (X_{1}+ \cdots + X_{r+1})^{2n} = 
\frac{r+1}{n} \sum_{k=1}^{n} \binom{2n}{2k} k \mathbb{E} X_{1}^{2k} 
\mathbb{E} (X_{2}+ \cdots + X_{r+1})^{2n-2k}.
\label{mess-1}
\end{equation}
\end{lemma}
\begin{proof}
Denote $X = X_{1}$ and $Y = X_{2} + \cdots + X_{r+1}$ and use the binomial 
theorem to obtain 
\begin{eqnarray*}
\mathbb{E} X(X+Y)^{2n-1} &  = & 
\sum_{k=0}^{2n-1} \binom{2n-1}{k} \mathbb{E}X^{k+1} \mathbb{E} Y^{2n-1-k} \\
& = & 
\sum_{k=0}^{n} \binom{2n-1}{2k-1} \mathbb{E}X^{2k} \mathbb{E}Y^{2n-2k} + 
\sum_{k=0}^{n-1} \binom{2n-1}{2k} \mathbb{E}X^{2k+1} \mathbb{E}Y^{2n-1-2k}.
\end{eqnarray*}
\noindent
The second sum vanishes as it only includes odd moments. Therefore
\begin{equation*}
\mathbb{E} X(X+Y)^{2n-1}  = 
\sum_{k=0}^{n} \binom{2n-1}{2k-1} \mathbb{E}X^{2k} \mathbb{E}Y^{2n-2k}.
\end{equation*}

The identity
\begin{eqnarray*} 
\mathbb{E} X (X+Y)^{2n-1} & = & \mathbb{E} X_{1} (X_{1} + X_{2} + 
\cdots + X_{r+1})^{2n-1} \\
 & = & \mathbb{E} X_{2} (X_{1} + X_{2} + 
\cdots + X_{r+1})^{2n-1} \\
 & = & \ldots \\
 & = & \mathbb{E} X_{r+1} (X_{1} + X_{2} + 
\cdots + X_{r+1})^{2n-1}, 
\end{eqnarray*}
\noindent 
comes from the fact that all the random variables have the same 
distribution.  The elementary identity 
\begin{equation}
\binom{2n-1}{2k-1} = \frac{k}{n} \binom{2n}{2k}
\end{equation}
\noindent
yields
\begin{equation*}
(r+1) \mathbb{E}X (X+Y)^{2n-1} =
\mathbb{E}(X+Y)(X+Y)^{2n-1} =
\mathbb{E}(X+Y)^{2n}.
\end{equation*}
\noindent
This gives the result.
\end{proof}

\begin{theorem}
The polynomials $B_{n}^{(\nu)}(r)$ satisfy the recurrences 
\begin{equation}
\label{recu-1a}
B_{n}^{(\nu)}(r+1) = \frac{r+1}{n} 
\sum_{k=1}^{n} k \binom{n}{k} \frac{(\nu+1)_{n}}{(\nu+1)_{k} (\nu+1)_{n-k}}
B_{n-k}^{(\nu)}(r)
\end{equation}
\noindent
and 
\begin{equation}
\label{recu-1b}
B_{n}^{(\nu)}(r) = \sum_{k=1}^{n} \left[ \frac{k(r+1)}{n}-1 \right] 
\binom{n}{k} 
\frac{(\nu+1)_{n}}{(\nu+1)_{k} (\nu+1)_{n-k}} B_{n-k}^{(\nu)}(r).
\end{equation}
\end{theorem}
\begin{proof}
Assume $\{ X_{i} \}$ is a collection of independent random 
variables, identically
distributed with a centered beta distribution
\begin{equation}
f_{\nu}(x) = 
\begin{cases}
\frac{1}{B\left( \nu + \tfrac{1}{2}, \tfrac{1}{2} \right)}
(1 - x^{2})^{\nu - \tfrac{1}{2}}, & \quad \text{ for } |x| < 1 \\
0 & \quad \text{ otherwise}.
\end{cases}
\label{beta}
\end{equation}
\noindent 
The associated characteristic function is 
\begin{equation}
\varphi_{\nu}(z) = \mathbb{E} e^{\imath z X_{i}} = 
\sum_{n=0}^{\infty} \mathbb{E} X_{i}^{2n} \, \frac{\left(-1\right)^n z^{2n}}{(2n)!} = 
\tilde{J}_{\nu}(z)
\end{equation}
\noindent
in view of
\cite[Entry 9.1.20]{abramowitz-1972a}, \cite[Entry 8.411.8]{gradshteyn-2007a}
\[
J_{\nu}\left(z\right) = \frac{\left(z/2\right)^\nu}{\pi^{1/2}\Gamma\left(\nu+\frac{1}{2}\right)}\int_{0}^{1}\left(1-t^2\right)^{\nu-\frac{1}{2}}\cos\left(zt\right) dt.
\]
In particular 
\begin{equation}
\mathbb{E} X_{i}^{2n} = \frac{1}{(\nu+1)_{n}} 
\frac{ (2n)!}{n! 2^{2n}} \text{ and }
\mathbb{E} X_{i}^{2n+1} = 0.
\end{equation}
\noindent
The definition of $B_{n}^{(\nu)}(r)$ and 
\begin{equation}
\left[ \tilde{J}_{\nu}(z) \right]^{r} = 
\left( \mathbb{E}e^{\imath z X_{1}} \right)^{r} 
= \mathbb{E} e^{\imath z(X_{1}+\cdots + X_{r})}
\end{equation}
imply
\begin{equation}
\mathbb{E} (X_{1}+ \cdots + X_{r})^{2n} = 
\frac{1}{(\nu+1)_{n}} \frac{(2n)!}{n! 2^{2n}} B_{n}^{(\nu)}(r).
\end{equation}
\noindent
Replacing in \eqref{mess-1} yields
\begin{multline*}
\frac{1}{(\nu+1)_{n}} \frac{(2n)!}{n! 2^{2n}} B_{n}^{(\nu)}(r+1) = \\
\frac{r+1}{n} 
\sum_{k=1}^{n} \binom{2n}{2k} k 
\frac{1}{(\nu+1)_{k}} 
\frac{(2k)!}{k!2^{2k}} 
\frac{1}{(\nu+1)_{n-k}} 
\frac{(2n-2k)!}{(n-k)!2^{2n-2k}} 
B_{n-k}^{(\nu)}(r),
\end{multline*}
\noindent
that reduces to \eqref{recu-1a}.

The second recurrence is obtained by applying the binomial formula
\begin{eqnarray*}
\mathbb{E} (X_{1}+ \cdots + X_{r+1})^{2n} & = & 
\sum_{k=0}^{n} \binom{2n}{2k} \mathbb{E}X_{1}^{2k} 
\mathbb{E} (X_{2}+ \cdots + X_{r+1})^{2n-2k} \\
& = & 
\mathbb{E} (X_{2} + \cdots + X_{r+1})^{2n} + 
\sum_{k=1}^{n} \binom{2n}{2k} \mathbb{E}X_{1}^{2k} 
\mathbb{E} (X_{2}+ \cdots + X_{r+1})^{2n-2k},
\end{eqnarray*}
\noindent
to obtain
\begin{multline*}
\mathbb{E} (X_{2} + \cdots + X_{r+1})^{2n}  = 
\mathbb{E} (X_{1} + \cdots + X_{r+1})^{2n}  \\
 - \sum_{k=1}^{n} \binom{2n}{2k} \mathbb{E}X_{1}^{2k} 
\mathbb{E} (X_{2}+ \cdots + X_{r+1})^{2n-2k}.
\end{multline*}
\noindent
In terms of the polynomials $B_{n}^{(\nu)}(r)$, the previous relation becomes 
\begin{eqnarray*}
\frac{1}{(\nu+1)_{n}} \frac{(2n)!}{n! 2^{2n}} B_{n}^{(\nu)}(r) & = & 
\frac{1}{(\nu+1)_{n}} \frac{(2n)!}{n! 2^{2n}} B_{n}^{(\nu)}(r+1)  \\
& & - \sum_{k=1}^{n} \binom{2n}{2k} 
\frac{1}{(\nu+1)_{k}} 
\frac{(2k)!}{k! 2^{2k}} 
\frac{1}{(\nu+1)_{n-k}} 
\frac{(2n-2k)!}{(n-k)! 2^{2n-2k}} 
B_{n-k}^{(\nu)}(r) \\
& = & \frac{1}{(\nu+1)_{n}} 
\frac{(2n)!}{n! 2^{2n}} 
\left( B_{n}^{(\nu)}(r+1) - 
\sum_{k=1}^{n}  \binom{n}{k} \frac{(\nu + 1)_{n}}{(\nu+1)_{k} (\nu + 1)_{n-k}} 
B_{n-k}^{(\nu)}(r) \right). 
\end{eqnarray*}
\noindent
Using the expression for $B_{n}^{(\nu)}(r)$ from \eqref{recu-1a} yields
\eqref{recu-1b}.
\end{proof}

\begin{note}
The identity \eqref{first-relation} follows by replacing \eqref{recu-1b} into 
\eqref{recu-1a}.
\end{note}

\section{The Cholewinski connection}
\label{sec-cholewinski}

The normalized Bessel function \eqref{bess-nor00} is written as 
\begin{equation}
\tilde{I}_{\nu}(z) = \sum_{n=0}^{\infty} \frac{z^{2n}}{b_{2n}(\nu)}
\end{equation}
\noindent
where 
\begin{equation}
b_{2n}(\nu) = 2^{2n} \frac{n! (n+ \nu)!}{\nu!}.
\end{equation}
\noindent
In his treatise  \cite{cholewinski-1988a}, Cholewinski defines a modified binomial coefficient
\begin{equation}
\binom{b_{2n}\left(\nu\right)}{b_{2k}\left(\nu\right)}=
\frac{b_{2n}\left(\nu\right)}{b_{2n-2k}\left(\nu\right)b_{2k}\left(\nu\right)}
=\binom{n}{k}\frac{\nu!\left(\nu+n\right)!}
{\left(\nu+n-k\right)!\left(\nu+k\right)!}
\end{equation}
and a new convolution by 
\begin{equation}
\left(x*y\right)^{\alpha}
=\sum_{k=0}^{+\infty}\binom{b_{\alpha}\left(\nu\right)}{b_{2k}
\left(\nu\right)}x^{2k}y^{\alpha-2k}=
y^{\alpha} \pFq21{- \tfrac{\alpha}{2}, - \tfrac{\alpha}{2} - \nu}{\nu+1}
{\frac{x^{2}}{y^{2}}}.
\end{equation}
\noindent 
Here ${_{2}F_{1}}$ is the classical hypergeometric function. 

In this notation 
\begin{equation}
\left[ \tilde{I}_{\nu}(z) \right]^{r} = 
\sum_{n=0}^{\infty} B_{n}^{(\nu)}(r) \frac{z^{2n}}{b_{2n}}.
\label{notation-cho}
\end{equation}

The next result gives a recurrence for a normalization of 
the polynomials $B_{n}^{(\nu)}\left(r\right)$ in Cholewinski's notations.

\begin{theorem}
The sequence of normalized polynomials 
\begin{equation}
\tilde{B}_{n}^{(\nu)}(r) = 
\frac{(2n)!}{n! 2^{2n} (\nu + 1)_{n}} B_{n}^{(\nu)}(r)
\label{cho-11}
\end{equation}
\noindent
is of binomial type; that is, they satisfy 
\begin{equation}
\tilde{B}_{n}^{(\nu)}(r+s) = \sum_{k=0}^{n} \binom{n}{k} 
\tilde{B}_{k}^{(\nu)}(r)
\tilde{B}_{n-k}^{(\nu)}(s).
\end{equation}
\noindent
The (unnormalized) polynomials $B_{n}^{(\nu)}(r)$ satisfy the identity
\begin{equation}
B_{n}^{(\nu)}(r+s) = \sum_{k=0}^{n} \binom{b_{2n}}{b_{2k}} 
B_{k}^{(\nu)}(r)
B_{n-k}^{(\nu)}(s).
\end{equation}
\end{theorem}
\begin{proof}
Start with 
\begin{eqnarray*}
\left[\tilde{I}_{\nu}\left(z\right)\right]^{r+s} & = & 
\left[\tilde{I}_{\nu}\left(z\right)\right]^{r}
\left[\tilde{I}_{\nu}\left(z\right)\right]^{s}\\
 & = & \sum_{p,q}B_{p}^{\left(\nu\right)}\left(r\right)
\frac{z^{2p}}{b_{2p}}B_{q}^{\left(\nu\right)}
\left(s\right)\frac{z^{2q}}{b_{2q}}\\
 & = & \sum_{n=0}^{\infty} \left(\sum_{p=0}^{n}\frac{b_{2n}}{b_{2p}b_{2n-2p}}
B_{p}^{\left(\nu\right)}\left(r\right)
B_{n-p}^{\left(\nu\right)}\left(s\right)\right)\frac{z^{2n}}{b_{2n}}.
\end{eqnarray*}
The result now follows from \eqref{notation-cho}.

\smallskip

An alternative proof of \eqref{cho-11} uses the moment representation
\begin{equation}
B_{n}^{\left(\nu\right)}\left(r\right)
=\left(\nu+1\right)_{n}\frac{n!}{2n!}2^{2n}
\mathbb{E}\left(X_{1}+\dots+X_{r}\right)^{2n}
\end{equation}
where $\left\{ X_{i}\right\} $ are independent, identically distributed
random variables with probability density \eqref{beta}. Simply observe that 
\begin{eqnarray*}
B_{n}^{\left(\nu\right)}\left(r+s\right) 
& = & \left(\nu+1\right)_{n}\frac{n!}{(2n)!}2^{2n}
\mathbb{E}\left(X_{1}+\dots+X_{r}+X_{r+1}+\dots+X_{r+s}\right)^{2n}\\
 & = & \left(\nu+1\right)_{n}\frac{n!}{(2n)!}2^{2n}\mathbb{E}
\left( \left(X_{1}+\dots+X_{r} \right)+ \left(Y_{1}+\dots+Y_{s} 
\right)\right)^{2n}
\end{eqnarray*}
\noindent
with $Y_{j} = X_{r+j}$ and use the binomial theorem. 
\end{proof}

\begin{theorem}
The polynomials $B_{n}^{(\nu)}(r)$ are given by 
\begin{equation}
B_{n}^{(\nu)}(r) = 
\frac{\mathbb{E} (X_{1}+ \cdots + X_{r})^{2n}}{\mathbb{E}(X_{1})^{2n}}.
\end{equation}
\noindent
The normalized polynomials are simply given by 
\begin{equation}
\tilde{B}_{n}^{(\nu)}(r) =
\mathbb{E} (X_{1}+ \cdots + X_{r})^{2n}.
\end{equation}
\end{theorem}
\begin{proof}
The moment representation 
\begin{equation}
B_{n}^{(\nu)}(r) = (\nu+1)_{n} \frac{n!}{(2n)!} 2^{2n} 
\mathbb{E} (X_{1}+ \cdots + X_{r})^{2n}
\end{equation}
\noindent 
is simplified using
\begin{equation}
B_{n}^{(\nu)}(1) = 1 = (\nu+1)_{n} \frac{n!}{(2n)!} 2^{2n} 
\mathbb{E}(X_{1})^{2n}.
\end{equation}
\end{proof}

\begin{note}
The integrals
\begin{equation}
W_{n}\left(s\right)=
\int_{\left[0,1\right]^{n}} 
\left\vert\sum_{k=1}^{n}e^{2\pi \imath x_{k}} \right\vert^{s} d\mathbf{x}
\end{equation}
have recently appeared in \cite{borweinj-2010c} in the study of short 
random walks in the plane. Their generating function is given by
\begin{equation}
\sum_{s=0}^{\infty} 
W_{n}(2s)\frac{(-z)^{s}}{s!s!}=\left[J_{0}(2\sqrt{z})\right]^{n}.
\end{equation}
\noindent
The relation
\begin{equation}
J_{0}(2\sqrt{z})=\tilde{J}_{0}(2\sqrt{z})
\end{equation}
and
\begin{equation}
\left[\tilde{J}_{0}(2\sqrt{z})\right]^{n}=
\sum_{s=0}^{\infty}B_{s}^{(0)}(n)\frac{(\imath\sqrt{z})^{2s}}{s!s!},
\end{equation}
give the link between the integrals $W_{n}(2s)$ and the polynomials 
$B_{n}^{\nu}\left(r\right)$ as 
\begin{equation}
W_{n}(2s)=B_{s}^{(0)}(n).
\end{equation}
\end{note}

\section{Some integer sequences}
\label{sec-integers}

The sequences $\{ a_{k} \}$ and $\{ M_{k} \}$ are defined there in 
\cite{bender-2003a} by the generating function 
\begin{equation}
x \left( \frac{\sqrt{x} I_{0}(2 \sqrt{x}) }{I_{1}(2 \sqrt{x})} - 2 \right) =
\sum_{k=1}^{\infty} a_{k}x^{k}
\end{equation}
\noindent
and 
\begin{equation}
M_{k} = (-1)^{k+1} (k+1)! k! a_{k+1}.
\end{equation}
\noindent 
In particular, it is stated that 
$\left\{ M_{k} \right\}$ is an integer sequence. This is proved next. An 
expression for $M_{k}$ in terms of the Bessel zeta function is given first.

\begin{theorem} 
The coefficients $M_{k}$ are given by 
\begin{equation}
M_{k} = k! (k+1)! 2^{2k} \zeta_{1}(2k).
\end{equation}
\end{theorem}
\begin{proof}
The identity
\begin{equation}
\label{real-bessel1}
\frac{d}{dz} I_{\nu+1}(z) = I_{\nu}(z) - \frac{\nu+1}{z} I_{\nu+1}(z),
\end{equation}
\noindent
implies 
\begin{equation}
\frac{d}{dz} \log I_{\nu+1}(z) = 
\frac{I_{\nu}(z)}{I_{\nu+1}(z)} - \frac{\nu+1}{z}.
\end{equation}
\noindent
Then \eqref{prod-111} gives 
\begin{equation}
\frac{d}{dz} \log I_{\nu+1}(z) = \frac{\nu+1}{z} + 
2 \sum_{p=1}^{\infty} (-1)^{p+1} \zeta_{\nu+1}(2p) z^{2p-1},
\end{equation}
\noindent
and \eqref{real-bessel1} then produces 
\begin{equation}
\frac{I_{\nu}(2 \sqrt{x})}{I_{\nu+1}(2 \sqrt{x})} = \frac{\nu+1}{\sqrt{x}}
+ 2 \sum_{p=1}^{\infty} (-1)^{p+1} \zeta_{\nu+1}(2p) (2 \sqrt{x})^{2p-1}.
\end{equation}
\noindent
Therefore 
\begin{equation}
\label{nu=0}
x \left( \frac{\sqrt{x} I_{\nu}(2 \sqrt{x})}{(\nu+1) I_{\nu+1}(2 \sqrt{x})} - 
2 \right) = -x + 
\frac{1}{\nu+1} 
\sum_{\ell=1}^{\infty} (-1)^{\ell+1} \zeta_{\nu+1}(2 \ell) 2^{2 \ell} 
x^{ \ell +1}.
\end{equation}
\noindent
The special case $\nu=0$ gives the result. 
\end{proof}

\begin{corollary}
The number $M_{k}$ is an integer.
\end{corollary}
\begin{proof}
The sequence 
\begin{equation}
\tilde{a}_{n} = 2^{2n+1} (n+1)! (n-1)! \zeta_{1}(2n)
\end{equation}
\noindent
has been shown to be an integer sequence in
\cite{amdeberhan-2013e}. The identity
\begin{equation}
M_{n} = \frac{n \tilde{a}_{n}}{2}
\end{equation}
\noindent
shows that $M_{n} \in \mathbb{N}$. Indeed, this is clear if $\tilde{a}_{n}$ 
is even. For $\tilde{a}_{n}$ odd, it must be 
that 
$n = 2(2^{m}-1)$ (see Theorem $8.1$ in \cite{amdeberhan-2013e}) and the 
result follows.
\end{proof}

\begin{note}
The recurrence
\begin{equation}
M_{n} = \sum_{r=1}^{n-1} \frac{\binom{n}{r}^{2}}{(r+1)(n-r+1)}M_{r}M_{n-r}
\end{equation}
\noindent 
with initial conditions $M_{1}=M_{2}=1$ follows from the recurrence for  
$\zeta_{1}(2n)$ given in \cite{amdeberhan-2013e}.
\end{note}

\medskip

At the end of \cite{bender-2003a}, the authors define the more general 
sequence $b_{n}(\nu)$, that has already appeared in \eqref{B-def}, by 
the generating function 
\begin{equation}
\label{def-bnu}
\sum_{n=1}^{\infty} \frac{x^{n}}{(\nu+n)! (n-1)!}b_{n}(\nu) = 
\frac{x}{(\nu+1)!}
\left( \frac{\sqrt{x}}{\nu+1}  \frac{I_{\nu}(2 \sqrt{x})}
{I_{\nu+1}(2 \sqrt{x})} - 
2 \right),
\end{equation}
\noindent
and it is stated that $b_{n}(\nu)$ is an integer sequence. This is incorrect, 
since it is easy to check that
\begin{equation}
b_{2}(\nu) = \frac{1}{\nu+1}.
\end{equation}

The next result presents a possible modification.

\begin{theorem}
The sequence $b_{n}(\nu)$ defined by \eqref{def-bnu} is given by
\begin{equation}
b_{n}(\nu) = \frac{(n+\nu)!(n-1)!}{(\nu+1)! (\nu+1)}(-1)^{n} 2^{2n-2} 
\zeta_{\nu+1}(2n-2), \text{ for } n \geq 2.
\end{equation}
\noindent
For $\nu \in \mathbb{N}$, the modified sequence 
\begin{equation}
\label{mod-seq1}
\tilde{b}_{n}(\nu) = (-1)^{n} b_{n}(\nu) 
\frac{(\nu+1)! (\nu+1)}{(\nu+n)! (n-1)!} \prod_{k=1}^{n} 
( k + \nu)^{\lf n/k \rf}, \text{ for } n \geq 2
\end{equation}
\noindent
takes integer values.
\end{theorem}
\begin{proof}
Comparing \eqref{nu=0} and \eqref{def-bnu} yields 
\begin{equation}
b_{n}(\nu) = \frac{(\nu+n)! (n-1)!}{(\nu+1)! (\nu+1)} (-1)^{n} 
\zeta_{\nu+1}(2n-2).
\end{equation}
\noindent
The function 
\begin{equation}
\phi_{2n}(\nu) = 2^{2n} \zeta_{\nu}(2n) 
\prod_{k=1}^{n} (\nu + k )^{\lf n/k \rf} 
\end{equation}
\noindent
has been shown to be a polynomial \cite{kishore-1964a} 
of degree $\begin{displaystyle}
1 - 2n + \sum_{k=1}^{n} \lf n/k \rf\end{displaystyle},$ with 
positive integer coefficients. This is the \textit{Rayleigh polynomial}. Then 
\begin{equation}
b_{\nu}(n) = \frac{(n+\nu)! (n-1)! (-1)^{n} \phi_{2n-2}(\nu+1)}
{(\nu+1)! (\nu+1) \prod_{k=1}^{n-1} (k+\nu+1)^{\lf (n-1)/k \rf}},
\end{equation}
\noindent
and the modified sequence \eqref{mod-seq1} reduces to $\phi_{2n-2}(\nu+1)$, 
giving the result.
\end{proof}

\medskip

\noindent
\textbf{Acknowledgments}. The second author acknowledges the partial 
support of NSF-DMS 1112656. The authors wish to thank A. Byrnes for her 
help. 


\end{document}